\documentclass[
singlecolumn,superscriptaddress]{revtex4-2}
\usepackage[utf8]{inputenc}
\usepackage{amsmath}
\usepackage{amsthm}
\usepackage{amssymb}
\usepackage{bbm}
\usepackage{braket}
\usepackage{tikz-cd}
\usepackage{tensor}
\newcommand{\of}[1]{\left( {#1}\right)}

\newtheorem{theorem}{Theorem}
\newtheorem{lemma}{Lemma}
\newtheorem{cor}{Corollary}

\begin{document}

\title{Sequencing the Entangled DNA of Fractional Quantum Hall Fluids}
\date{February 2022}

\author{Joseph R. Cruise}
\email{c.joseph@wustl.edu}
\affiliation{Department of Physics, Washington University in St. Louis, St. Louis, Missouri 63130, USA}

\author{Alexander Seidel}
\email{seidel@physics.wustl.edu}
\affiliation{Department of Physics, 
TFK, Technische Universit\"at M\"unchen, James-Franck-Straße 1, D-85748 Garching, Germany}
\affiliation{Department of Physics, Washington University, St. Louis, Missouri 63130, USA}

\begin{abstract}
We introduce and prove the ``root theorem'', which establishes a  condition for families of operators to annihilate all root states associated with zero modes of a given positive semi-definite $k$-body Hamiltonian chosen from a large class.
This class is motivated by fractional quantum Hall and related problems, and features generally long-ranged, one-dimensional, dipole-conserving terms.
Our theorem streamlines analysis of zero-modes in contexts where ``generalized'' or ``entangled'' Pauli principles apply. One major application of the theorem is to parent Hamiltonians for mixed Landau-level wave functions, such as unprojected composite fermion or parton-like states that were recently discussed in the literature, where it is difficult to rigorously establish a complete set of zero modes with traditional polynomial techniques. As a simple application we show that a modified $V_1$ pseudo-potential,  obtained via retention of only half the terms, stabilizes the $\nu=1/2$ Tao-Thouless state as the unique densest ground state.
\end{abstract}

\maketitle

\section{Introduction}

Recent decades have greatly added to our understanding that quantum phases of matter are not comprehensively described by the conventional, symmetry-based orders described by Landau\cite{Landau:1937obd}. Instead,
the description of quantum orders is a 
paramount challenge in the study of complex interacting many-body systems. While we may still be far from a complete understanding of all such possible orders, a time-tested approach toward adding to this understanding is identifying a ``correct'' quantum-many body wave function describing a given phenomenology, and extracting from it the essential features responsible for its universal behavior. Perhaps the first example of this approach is the BCS trial wave function for superconductivity\cite{BCS}, even though the insight of its Landau-paradigm defying character came somewhat later\cite{HANSSON2004497}.
Interacting  quantum many-body systems in one spatial dimension represent another important niche where this 
approach has been fruitful. While integrability provides a powerful tool here---albeit one peculiar to the one-dimensional situation---important general principles, e.g., regarding gapped, symmetry protected, topological phases have been discovered through explicit  construction of classes of many-body wave functions. The original case in point here is the seminal AKLT model\cite{Affleck1987}, which sparked vast efforts of generalization, and generation of methodology for classification via matrix product states, (e.g., \cite{Fannes1992b, Pollmann2010, Fidkowski2011, XChen2011}). For gapped, partially gapped, and gappless one-dimensional quantum-systems, another method reliant on the identification of correct trial wave functions is that of the ``factorization of the wave function''\cite{Ogata1990, Seidel2004a, Ribeiro2006, Kruis2004}. Indeed, in this context, there emerges some notion of ``squeezing''\cite{Kruis2004} that, perhaps not coincidentally, may be a close cousin of a similar notion that will be of central importance in the present work.

In this paper, our foremost concern will be the fractional quantum Hall effect and models for related phenomena. Here, it was of course Laughlin who pioneered the technique of writing down unique trial wave functions that capture the underlying physics. This has been key to the discovery of a plethora of quantum orders in the fractional quantum Hall regime. For many holomorphic, single-component lowest-Landau-level wave functions, \`a la Laughlin, a highly consistent framework exists to link those wave functions to field theory\cite{Moore1991}, in particular, the effective edge theory. 
The connection with universal physics at the edge can, however, be made in another very convincing way that involves a certain type of parent Hamiltonian. The prototype for this kind of Hamiltonian is the $V_1$ Haldane pseudopotential (or its bosonic counterpart)\cite{Haldane1983}, with an equivalent Trugman-Kivelson form\cite{Trugman1985}. This Hamiltonian stabilizes the Laughlin state as its unique lowest-angular-momentum zero-energy eigenstate (zero mode). At higher angular momenta, however, the zero-mode space is increasingly degenerate. Indeed, the counting of linearly-independent zero modes at a given angular momentum relative to the 
incompressible fluid exactly reproduces the mode counting in the edge conformal field theory. A number of  other such parent Hamiltonians exist that conform to such a ``zero mode paradigm''\cite{Haldane1983, Halperin1983, Trugman1985,Jain1990,Rezayi1991, Greiter1992, Read1996, Ardonne2001, Simon2007, Chen2017, Bandyopadhyay2018, Bandyopadhyay2019, Greiter2021,  2022arXiv220409684T}.
Originally, this was
especially the case for holomorphic, lowest Landau-level wave functions. In the presence of an appropriate confining potential that lifts the zero-mode degeneracy, and assuming a bulk gap, this rigorously establishes the edge physics in a microscopic model. Due to bulk-edge correspondence in these systems, one may argue that this is basically as good as knowing the bulk topological quantum field theory. Indeed, recent probes of the statistics of fundamental particles make essential use of the physics at the edge\cite{Nakamura2020}, as do other probes, such as momentum-resolved tunneling\cite{Huber2005, Seidel2009a, Wang2010}. 

Recently, there has been much renewed interest in wave functions that are not holomorphic, at least not without projection to the lowest Landau level.
The best known examples are the Jain composite fermion states\cite{Jain1989a}, but also beyond that a larger class of so-called ``parton states'', which historically were among the first non-Abelian fractional quantum Hall states to be written down\cite{Jain1989, Jain1990a, Wen1992}.  
The recently revived interest\cite{ Wu2017, Balram2018, Balram2018a, Bandyopadhyay2018, Balram2021, 2022arXiv220409684T} in parton states stems, in part, from the arrival of new platforms to harbor the fractional quantum Hall effect, such as graphene in its multi-layer varieties.
Not being described by holomorphic polynomials, the general framework of Ref. \cite{Moore1991} does not immediately apply, and indeed, one can argue that the identification of ``zero-mode paradigm conforming'' parent Hamiltonians is the most direct way to establish universal physics from microscopic principles. At the same time, where such Hamiltonians exist, their zero mode spaces are considerably harder to study than for Hamiltonians stabilizing holomorphic lowest Landau-level wave functions. This is so because techniques surrounding symmetric polynomials, which have, for example, been successfully used in the case of the Laughlin\cite{Stone1990} and several other\cite{Read1996, Milovanovic1996,Ardonne2001} quantum Hall states to facilitate such counting, do not apply in this case.
For reasons of this kind, parent Hamiltonians (with proper zero-mode counting) for Jain composite fermion states 
have only recently been constructed\cite{Bandyopadhyay2019}. On the other hand, non-Abelian parton states can, in many cases,  be characterized by analytic-clustering conditions that lead to natural parent Hamiltonians. However, the rigorous full characterization of the zero-mode spaces of such Hamitonians can, again, not be carried out with the polynomial techniques--or any standard device of commutative algebra--that are characteristic of the lowest Landau level. Instead, radically different methods for zero-mode characterization and counting have been developed\cite{Bandyopadhyay2018, Chen2018b, Bandyopadhyay2019, 2022arXiv220409684T} that put less emphasis on analytic wave functions and larger emphasis on second quantization, breaking with the tradition of the field. 
The key ingredient to these methods is the generalization of squeezing principles that had hitherto been associated with the lowest Landau level, especially Jack-polynomial wave functions\cite{Bernevig2008, Bernevig2008a}. These generalizations, leading to notions of ``root states'' governed by Pauli-like principles, are then turned into devices for the rigorous characterization of {\em complete} sets of zero modes. The resulting methods certainly also work for parent Hamiltonians of holomorphic wave functions, but moreover can be applied to non-holomorphic, mixed-Landau-level wave functions and their parent Hamiltonians as well, where we know of no alternative.

The notion of a root state will be central to our work.
We give a technical defintion in the following section, where we formalize generalizations for multi-component/multi-Landau-level states of the recent literature. For general intuition, let us appeal to the well-known root state of the $\nu=1/3$ Laughlin state. This can be represented by a string of consecutive numbers of a simple pattern,  $100100100100\dotsc$, see also equation \eqref{eq:pattern} below.
There are several ways to uniquely assign simple strings of patterns to the set of zero modes of certain Hamiltonians, including taking the thin torus/thin cylinder limit. In this work, we will emphasize squeezing principles. The aforementioned simple patterns do efficiently encode all of the edge physics, thus, as argued above, all of the universal physics. Indeed, this includes braiding statistics, which can be extracted from these patterns rather systematically and directly~\cite{Flavin2011}, as was recently demonstrated also for mixed-Landau level states ~\cite{Bandyopadhyay2018,  2022arXiv220409684T}. For these reasons, we think of simple patterns associated to root states as the ``DNA'' of the underlying quantum Hall state. Though many as aspects of the single component/single Landau level case carry over to the multi-component/multi-Landau level case, the structure of the root states becomes more intricate in the latter case as roots states, or DNA, will retain simple forms of entanglement. 
This makes the study of such root structure a much more interesting and generally more demanding problem.

The main purpose of this paper is to prove a theorem that we hope will streamline root-state analyses of this kind, with applications for composite fermion and parton states, as well as possibly other, even more general holomorphic and non-holomorphic many-body trial wave functions in mind.
It will simplify and unify the application of the underlying principles, and allow for generalization in several directions. In particular, the theorem will readily cover $k$-body interactions and generalizations of Pauli-like principles, whereas the non-holomorphic examples referenced above are all $2$-body cases. It will also allow questions concerning zero mode characterization of certain classes of Hamiltonians to be addressed without undue focus on case-by-case adaption of proof techniques.

The remainder of this paper is organized as follows.
In section \ref{statement} we establish the setting for our theorem and subsequently state it. In section \ref{applications} we first illustrate simple but typical use cases for our theorem, which we hope will convince the reader of its general utility.
In section \ref{proof} we prove the theorem. We conclude in section \ref{conclusion}.

\section{Statement of the problem\label{statement}}

Consider a (bosonic or fermionic) Fock space with a single particle basis
\begin{equation}\label{eq:SP} c^\dagger_{j,r} \ket{0}
\end{equation}
for $j\in \mathbb{N}$ and $r \in \{ 0,\dots,n-1\}$ for some $n \in \mathbb{N}$, where $\ket{0}$ is the vacuum. 
Pictorially, we may think of this single particle space as the quasi-one-dimensional ladder represented in figure \ref{Fig1a}. This is the setting appropriate to the physics of a single Landau level in disk geometry with internal spin or valley degrees of freedom, where $j$ represents the angular momentum of the orbital, and $r$ the internal index. 
\begin{figure}
    \centering
    \begin{tikzpicture}
    \filldraw[black] (0,0) circle (2pt);
    \filldraw[black] (1,0) circle (2pt);
    \filldraw[black] (2,0) circle (2pt);
    \filldraw[black] (3,0) circle (2pt);
    \filldraw[black] (4,0) circle (2pt);
    \filldraw[black] (0,1) circle (2pt);
    \filldraw[black] (1,1) circle (2pt);
    \filldraw[black] (2,1) circle (2pt);
    \filldraw[black] (3,1) circle (2pt);
    \filldraw[black] (4,1) circle (2pt);
    \filldraw[black] (0,2) circle (2pt);
    \filldraw[black] (1,2) circle (2pt);
    \filldraw[black] (2,2) circle (2pt);
    \filldraw[black] (3,2) circle (2pt);
    \filldraw[black] (4,2) circle (2pt);
    
    \draw[->] (-0.05,-0.8) -- (6.05,-0.8);
    \draw[->] (-0.7,-0.1) -- (-0.7,2.1);
    
    \node at (3,-1.4) {\large Angular momentum $j$};
    \node[rotate = 90] at (-1.3,1) {\large Internal index $r$};
    \node at (5.5,1) {\Huge \dots};
    
    \end{tikzpicture}
    \caption{Half infinite ladder of single particle quantum Hall states with simple $j\geq 0$ boundary condition, as appropriate, e.g., for a single Landau level in the disk geometry with an additional spin or valley index, or multiple Landau levels in the half-infinite cylinder.}
    \label{Fig1a}
\end{figure}
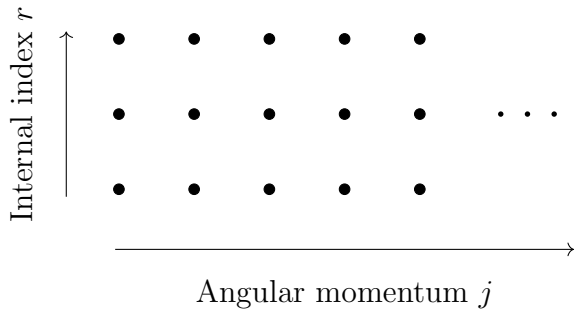
It is equally appropriate to a system confined to a finite number of multiple Landau levels. 
This is immediately natural for a half-infinite cylinder geometry with multiple Landau levels. Alternatively, in disk geometry, $j$ may represent the sum of the angular momentum and the Landau-level index.
This would likewise reproduce the ``shape'' of the half-infinite ladder of figure  \ref{Fig1a} (see, e.g., Ref. \cite{Bandyopadhyay2019}).
Independent of context, we will refer to $j$ as the single-particle angular momentum in this paper.
In addition, slightly different boundary conditions  may be appropriate in disk geometry:
many works in the literature, to which our findings will also directly apply, represent $n$ Landau levels as a half-infinite ladder system via the boundary condition shown in figure  \ref{Fig1b}. In this case, $j$ still represents \textit{bona fide} angular momentum and runs over all integers $j>-n$, and we enforce the constraint $r\geq -j$, leading to the ``shape'' shown in the figure. Neither the inclusion of negative $j$ 
nor that of additional constraints of the type mentioned (including a possible right boundary as appropriate in spherical geometry) will cause any problems in the following, so we may restrict the setting to that of figure  \ref{Fig1a} without loss of generality, when definiteness is required.
\begin{figure}
    \centering
    \begin{tikzpicture}
    \filldraw[black] (2,0) circle (2pt);
    \filldraw[black] (3,0) circle (2pt);
    \filldraw[black] (4,0) circle (2pt) ;
    \filldraw[black] (1,1) circle (2pt);
    \filldraw[black] (2,1) circle (2pt);
    \filldraw[black] (3,1) circle (2pt);
    \filldraw[black] (4,1) circle (2pt);
    \filldraw[black] (0,2) circle (2pt);
    \filldraw[black] (1,2) circle (2pt);
    \filldraw[black] (2,2) circle (2pt);
    \filldraw[black] (3,2) circle (2pt);
    \filldraw[black] (4,2) circle (2pt);
    
    \draw[->] (-0.05,-0.8) -- (6.05,-0.8);
    \draw[->] (-0.7,-0.1) -- (-0.7,2.1);
    
    \node at (3,-1.4) {\large Angular momentum $j$};
    \node[rotate = 90] at (-1.3,1) {\large Internal index $r$};
    \node at (5.5,1) {\Huge \dots};
    
    \end{tikzpicture}
    \caption{Half infinite ladder of quantum Hall states for a system of multiple Landau levels in the disk geometry, subject to the constraint $r \geq -j$. }
    \label{Fig1b}
\end{figure}
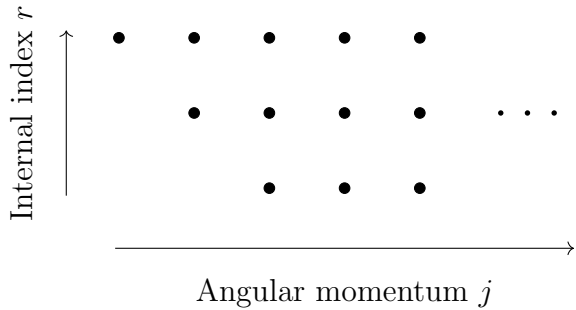

Our focus is on positive (semi-)definite $k$-body Hamiltonians of the form
\begin{equation}\label{eq:Hl} H =  \sum_{p=1}^k H_p =\sum_{p=1}^k\of{ \sum_{q=1}^{m_p}  \sum_J \of{Q_J^{p,q}}^\dagger Q_J^{p,q}}, \end{equation}
where 
\begin{equation} Q_J^{p,q} = \sum_{j_1,\dots,j_p}\sum_{r_1,\dots,r_p} \delta_{j_1+\dots+j_p,J} \eta^{p,q}_{j_1,\dots,j_p,r_1,\dots,r_p} c_{j_1,r_1}\dots c_{j_p,r_p}\;,\label{eq:Q} \end{equation} 
with $q$ an additional label that runs over all $m_p$ $p$-body terms present in the Hamiltonian.
This is the form one generally obtains for pseudo-potential expansions of positive (semi-definite) interactions in second quantization. Such expansions were originally introduced in first quantization by Haldane\cite{Haldane1983}, with the $2$-body, single Landau-level case in mind, but have been extensively generalized since. 
We emphasize that, while in many cases of interest Hamiltonians of this type will be local in real space, the operators $Q_J^{p,q}$ will generally be long-ranged in the Landau-level basis used here. This represents a considerable obstacle both in analytic and numerical treatments. 
While various symmetries may be imposed, here we will assume only angular momentum conservation, as enforced by the Kronecker-delta in equation \eqref{eq:Q}. This renders the model center-of-mass- \cite{Seidel2005}, or, in modern language, dipole-conserving.

Although equation \eqref{eq:Hl} is of a very general form, solvable models of this type are usually frustration free: known eigenstates will be zero modes, i.e., ground states that are annihilated by the Hamiltonian. 
The positive semi-definiteness of each term $\of{Q_J^{p,q}}^\dagger Q_J^{p,q}$ then  necessitates
\[ Q_J^{p,q}\ket{\psi} = 0, \]
for all $p,$ $q$, and $J$, i.e., any zero mode $\ket{\psi}$ will be annihilated by each such term in the Hamiltonian.

Consider now the basis states of the Fock space of the form
\begin{equation}\label{eq:S} \ket{S} = \ket{j_1r_1,\dots,j_Nr_N} = c^\dagger_{j_1,r_1}\dots c^\dagger_{j_N,r_N} \ket{0}, \end{equation}
where the ``$S$'' is a reference to Slater-determinants. We specialize to fermions in nomenclature only. Indeed, our results will equally hold for bosons, except where otherwise noted, where the state $\ket{S}$ will be a bosonic permanent, also know as a ``monomial'' in the literature, which is short for ``symmetrized monomial''. Instead of Slater-determinant or permanent/monomial, we will often use the neutral term {\em configuration} to refer to such occupation number eigenstates. 

We define the \textit{pattern of} the configuration $\ket{S}$ to be the sequence
\[ \{n_j\} = \{ \braket{S|\hat{n}_j|S} \}, \]
where
\begin{equation} \hat n_j = \sum_r c_{j,r}^\dagger c_{j,r}
\label{eq:number}\end{equation} 
is the usual number operator. Necessarily
\[ \sum_j n_j = N \text{ and } n_j \geq 0\]
for each $j$. For $\ket{\psi}$ an $N$-particle zero-mode, we will call a sequence $\{n_j\}$ 
a {\em pattern of} $\ket{\psi}$ if a basis state $\ket{S}$ of this pattern has a non-zero coefficient in the expansion of $\ket{\psi}$.

Any pattern can be represented either by specifying the sequence $\{n_j\}$ as introduced above, or by specifying the sequence $j_1,\dotsc, j_N$ of the underlying state \eqref{eq:S}. (The $r$-indices defining this state are obviously irrelevant to the pattern.)
We will denote such sequences, defining patterns in terms of $j$ (``angular momentum'') quantum numbers, by $[j_i]$, to distinguish them from equivalent but distinct occupancy-number sequences $\{n_j\}$. As, in particular, the elements of $[j_i]$ add to the total angular momentum of the state, $J=\sum_i j_i$, it is natural to view the associated pattern as a partition of the number $J$.\footnote{In mathematics, the terms $j_i$ of the partition are conventionally non-negative, whereas in some physical contexts negative $j_i$ are permitted, see above. This can always be remedied via positive additive shifts, which will, however, be without consequence in the following. }
In this context, we think of the $j_i$ as the ``terms'' or ``parts'' of the partition, and the $n_j$ as the ``multiplicities'', where $n_j$ is the number of times the part $j$ occurs in the partition. We may always go back and forth between these representations via
\[
n_j = |\{ i: j_i=j\}|\quad\text{and}\quad
j_i= \text{min}\{ j: \sum_{j'=j_{\sf min}}^j n_{j'} \geq i \}\,,
\]
where $j_{\sf min}$ is the minimum allowed $j$-value, and, in writing the above, we have adopted the convention $j_1\leq j_2\leq\dotsc\leq j_N$ (which differs from most mathematical texts on partitions). We will use this convention unless otherwise noted.

For patterns/partitions corresponding to the same $J$, we may adopt the partial order, or ``dominance relation'' found in the mathematical literature, where
\begin{align}
[j_i]\geq {[j'_i]} \;\; \iff\;\; \sum_{\substack{i\\ i\geq i_0}} j_i \geq 
\sum_{\substack{i\\ i\geq i_0}} j'_i
 \quad \text{for all}\;\;i_0\,.
\end{align}
Whenever $[j_i]\geq {[j'_i]}$, we will also write $\{n_j\}\geq \{n'_j\}$ for the occupancies/multiplicites associated with $[j_i]$, ${[j'_i]}$, respectively.
We will, in particular, be interested in applying this dominance relation to root patterns. 
Given that we only consider boundary conditions that render any subspace of a given $J$ finite dimensional, there are always maximal elements under this relation. We will call any pattern of a state $\ket{\psi}$ that is not dominated by any of the other patterns of $\ket{\psi}$ a {\em root pattern} of $\ket{\psi}$. As stated, if $\ket{\psi}$ has well-defined angular momentum $J$, which we will always assume, there must always be at least one root pattern. Whenever there is only a single Landau level without internal degrees of freedom, this definition of a root pattern/root partition agrees with that found in the literature(e.g. \cite{Bernevig2008, Mazaheri2015}). 
For the case of multiple Landau levels and/or internal degrees of freedom, we obtain patterns from configurations by ``forgetting'' the internal $r$-indices, or ``collapsing'' along the $r$-direction.
We may thus think of a pattern as a configuration of identical particles occupying a single Landau level, where multiple occupancies are allowed (irrespective of whether the original particles {\em with} the $r$-degree of freedom were fermions or bosons).  
This is represented in figure  \ref{Fig2}.
In terms of the occupancies $\{n_j\}$, we may then understand dominance as follows\cite{macdonald1998symmetric}. Consider ``inward squeezing'' processes as shown in the figure, where a pair of particles is hopping inward while conserving total $J$.
Then, $\{n_j\}\geq \{n'_j\}$ if and only if $\{n'_j\}$ can be obtained from  $\{n_j\}$ via a (finite, possibly empty) sequence of inward squeezing processes. 

Having now a suitable definition of the root patterns of a state $\ket{\psi}$ that is applicable to the general situation of multiple Landau levels/internal degrees of freedom, we proceed by defining a similarly general notion of a {\em root state}.
Given any root pattern $\{n_j\}$ of a state $\ket{\psi}$, we will refer to any configuration $\ket{S}$ that has pattern $\{n_j\}$ {\em and} appears in the expansion of $\ket{\psi}$ into occupation number eigenstates as a {\em root configuration} of $\ket{\psi}$. Equivalently, we will refer to the sequence of multi-indices $(j_i,r_i)$ defining $\ket{S}$ in \eqref{eq:S} as a root configuration.



For any root pattern $\{n_j\}$, we also define an associated projection operator
\begin{equation} \label{eq:pn} P_{\{n_j\}} := \sum_{\ket{S},\,\ket{S}\;\text{has pattern}\;\{n_j\}  } \frac{1}{\braket{S|S}} \ket{S}\bra{S}.  \end{equation}
It is generally possible for a zero mode to have multiple root patterns.
 We may thus also consider the entire subspace of all root configurations by defining the \textit{root projection operator}, 
\begin{equation} \label{eq:Proot}P_\text{root} := \sum_{\{n_j\} \text{ a root pat. of } \ket{\psi}} P_{\{n_j\}}. 
\end{equation}  
We then refer to $P_\text{root}\ket{\psi}$ as the \textit{root state}, or equivalently, the DNA of $\ket{\psi}$. (This definition of a root state is slightly more restrictive than that used in most of the literature with $n>1$, but will agree in practical applications.) 
The root states consists precisely of that part of the expansion of $\ket{\psi}$ into basis states $\ket{S}$
that one obtains if only root configurations are kept.
Note that in equation \eqref{eq:pn}, we chose to sum over all $\ket{S}$ of a given root pattern.
By definition, these need not all be root configurations, as some of them (but not all) may not appear in the expansion of 
$\ket{\psi}$: in $P_\text{root}\ket{\psi}$, these extra terms do not matter. In this way, one obtains a rather simple form for
$P_{\{n_j\}}$ in terms of second quantized ladder operators,

\begin{equation} \label{eq:Pn}
P_{\{n_j\}} = \frac{1}{\prod_j (n_j)!} \sum_{r_1,\dots,r_N} c^\dagger_{j_1 r_1} \dots c^\dagger_{j_N,r_N}c_{j_N,r_N}\dots c_{j_1,r_1}, \end{equation}
where the sum goes over all sequences of multi-indices with pattern $\{n_j\}$.
Strictly speaking, the two forms of the projection operator are equivalent while acting on states of up to $N$ particles only. However, we will have no use for the operator $P_{\{n_j\}}$ to act on states with more than $N$ particles. We finally emphasize that though we notationally suppress it, the definition of $P_{\text{root}}$ explicitly depends on $\ket{\psi}$, and in expressions like $P_{\text{root}}\ket{\psi}$, it will be generally implied that $\ket{\psi}$ is acted upon by ``its root projector''. As such, ``taking the root state'' is thus by no means a linear operation, however, for {\em fixed} $\ket{\psi}$, Equation \eqref{eq:Proot} is, of course, that of a linear operator. For brevity, we will also write $\ket{\psi}_{\text{root}}$ for the root state of $\ket{\psi}$.

\begin{figure}
    \centering
    

    \includegraphics[width=5 in]{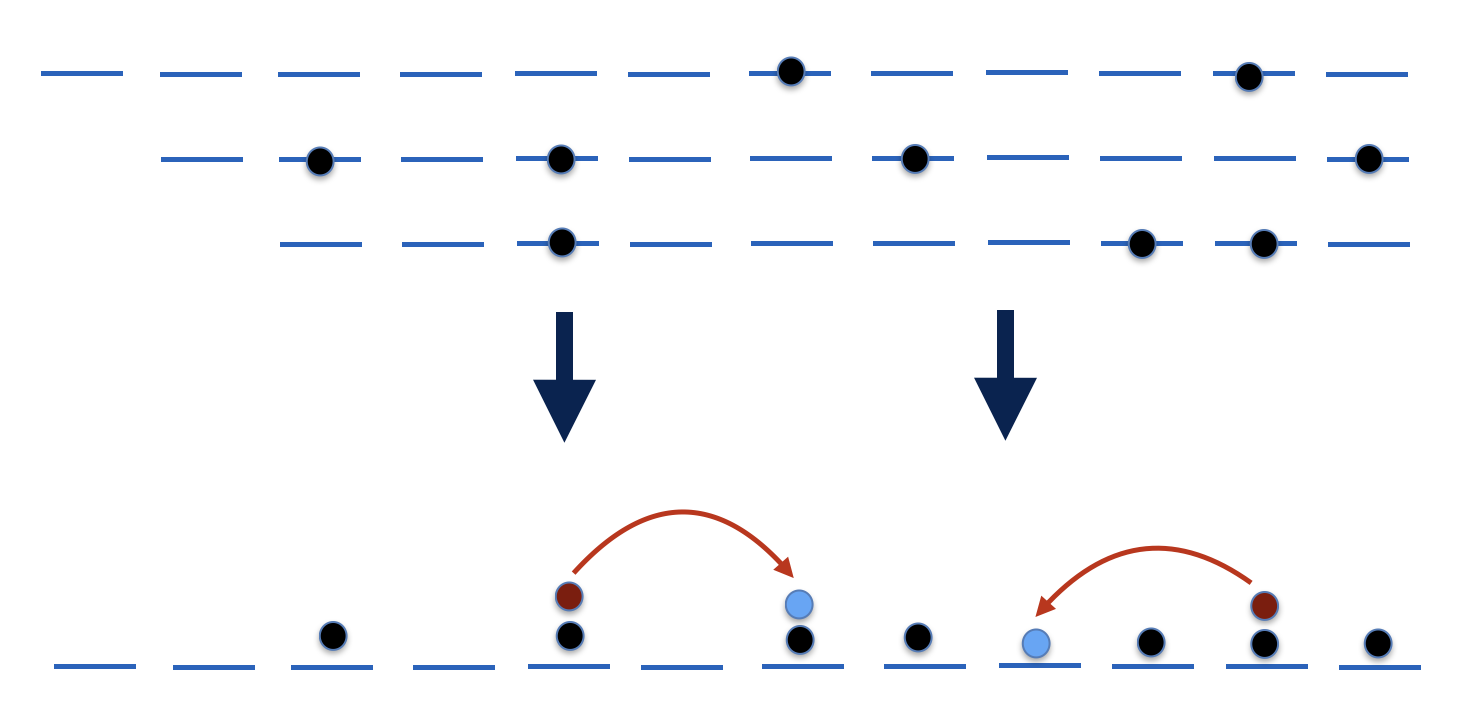}
    \caption{Graphical representation of patterns/partitions (bottom) associated to configurations (top).
    The pattern/partition is obtained by dropping the vertical index ($r$) present in the configuration, and representing each horizontal index ($j$) by a single orbital. These single orbitals are then occupied by identical particles that may multiply occupy each orbital, resulting in what we refer to as a pattern of occupancies $n_j$ (see text). (If the original particles were bosons, multiple occupancies are already possible for each state represented by a horizontal line in the configuration on top.) On patterns, one may define inward squeezing processes as indicated by the arrows and as explained in the text.
    }
    \label{Fig2}
\end{figure}

We now have everything in place to state the main result of this work.

\begin{theorem}\label{rootthm}
For fixed $(j_1,\dotsc ,j_l)$, let $Z$ be an $l$-body operator of the form 
\begin{align}\label{eq:Z} Z &:= \sum_{r_1,\dots r_l} \omega_{r_1,\dots,r_l} c_{\beta_{1}} \dots c_{\beta_{l}},
\end{align}
with $\beta_i = (j_i, r_i)$. If 
\[ Z \ket{\psi^l}_{\sf root} = 0 \]
for all $l$-particle zero-modes $\ket{\psi^l}$ of the Hamiltonian $H$ given in equation \eqref{eq:Hl}, then 
\begin{equation} \label{eq:ZZeroN}  Z\ket{\psi}_{
\sf root}= 0 \end{equation}
for any zero-mode $\ket{\psi}$ of $H$, regardless of particle number. 
\end{theorem}

Before proving this theorem, we first demonstrate its utility with a few examples, some well-established and some lesser known.

\section{Some basics applications\label{applications}}
\subsection{Laughlin state}
The previous theorem
is acutely relevant to situations involving multiple-component states, such as  the discussion of zero-mode properties of Refs. \cite{Seidel2008a, Seidel2011, Chen2017, Bandyopadhyay2018,  Bandyopadhyay2019, 2022arXiv220409684T}.
In particular, it furnishes a streamlined machinery for the discussion of zero modes, and generalizes
previous approaches
to three- and higher-body interactions. That said, it retains its usefulness in the simpler context of frustration free interactions in single component states; to demonstrate the utility of the theorem, we first turn to the simplest such example: the label $r$ can then be dropped.

Perhaps the best-known single component state is the (fermionic) Laughlin state at filling factor $1/3$, with first quantized wave function,
\begin{equation}\label{eq:psiL}
\psi_L = \prod_{i<j} (z_i-z_j)^3 \,\exp(-\frac 14 \sum_k |z_k|^2)\,,    
\end{equation}
in disk geometry. Here $z_i=x_i+iy_i$ is the complex coordinate of the $i$th particle.
In a seminal paper\cite{Rezayi1994}, Rezayi and Haldane observed that (in the cylinder geometry, though this makes no difference to the root analysis) this wave function has a ``Tao-Thouless''\cite{Tao1983} type root pattern of the form
\begin{equation}\label{eq:pattern}
  \{n_j\}=  1001001001001001\dotsc ,
\end{equation}
which is the aforementioned DNA of the Laughlin state. Note that for single-component states, a root pattern fully determines the associated root state, assuming the root pattern is unique. 
The Laughlin state equation \eqref{eq:psiL} is the densest (i.e., lowest in terms of angular momentum) zero mode of the $V_1$ Haldane pseudo-potential, which, in first quantization, takes on the form
\begin{equation}\label{eq:V1}
    \hat V_1 = \sum_{i<j} \hat P_{ij}\,,
\end{equation}
where $P_{ij}$ projects the pair of particles with indices $i$ and $j$ onto the subspace where this pair has relative angular momentum $1$.
Similar patterns as in 
equation \eqref{eq:pattern} apply to all root states of zero modes: any pattern satisfying  the constraint of having no more than one particle in any three adjacent sites occurs as the root pattern of some zero mode. This can, for example, be seen from a thin cylinder analysis \cite{Rezayi1994, Seidel2005}, and constraints of this kind were first appreciated to apply to a large class of states of Jack-polynomial type\cite{Bernevig2008, Bernevig2008a}, and have been coined ``generalized Pauli principles'' (GPPs).
We will now re-derive some of these results using our theorem. 
Accordingly, we first express $\hat V_1$ in second quantization, which is well-known to be of the following form,
\begin{equation}\label{eq:V12nd}
\hat V_1
= \sum_{J} Q_J^\dagger Q_J\,,
\end{equation}
where the precise expression for the $Q_J$ depends somewhat on the geometry and was perhaps first explicitly stated for the cylinder/torus case\cite{Lee2004}. Here, we will consider the disk geometry for definiteness, where
$J\geq 0$ and\cite{Ortiz2013}
\begin{equation}\label{eq:QJ}
Q_J = \sum_{\substack{p>0\\
p\in \mathbb{Z} \;\text{for}\;J\;\text{even}\\
p\in \mathbb{Z}+
\frac 12 \; \text{for}\;J\;\text{odd}}} p2^{-J/2+1}\sqrt{\frac{2}{J} { J \choose J/2+p}} c_{\frac J2-p}c_{\frac J2+p}.
\end{equation}
$\hat V_1$ thus satisfies the requirements needed to apply our theorem. While the proof of the theorem, given below, heavily utilizes the second quantization of the Hamiltonian, one efficient way to proceed is to study
$k=2$-particle zero modes using the original first-quantized form of the Hamiltonian. Note that 
\begin{equation}
    \psi_{m, J} = (z_1-z_2)^m (z_1+z_2)^{(J-m)}\,,
\end{equation}
is a basis of all two-particle states in the lowest Landau level,
where $m$ is odd and $J\geq m$, and we drop the obligatory Gaussian factor displayed in equation \eqref{eq:psiL}.
Here, $m$ and $J$ play the roles of relative and total angular momentum of the 2-particle state, respectively. A complete set of 2-particle zero modes is thus given by the $\psi_{m,J}$ with $m>1$, or $m\geq3$. By forming suitable new linear combinations, we arrive at the alternative basis
\begin{equation} 
\phi_{m,J}=(z_1-z_2)^m\times
\begin{cases} 
      (z_1z_2)^{\frac{J-m}{2}} & J\;\text{odd} \\
  z_1^{\frac{J-m-1}{2}} z_2^{\frac{J-m+1}{2}} +
   z_1^{\frac{J-m+1}{2}} z_2^{\frac{J-m-1}{2}}& J\;\text{even}   
   \end{cases}\;.
\end{equation}
Expressing the product $z_1z_2$ in terms of $z_1+z_2$ and $z_1 - z_2$, one sees that the $\phi_{m,J}$ are linear combinations of the $\psi_{m',J}$ with same $J$ and $m'\geq m$. Therefore, the matrix relating $\phi$'s to $
\psi$'s being triangular, the $\phi_{m,J}$ with $m\geq 3$ still are a basis of the 2-particle zero mode space. Since $c_r^\dagger$ creates a state proportional to $z^r$, the expansion of $\phi_{m,J}$ into 2-particle Slater determinants $\ket{S}$ can be read off the monomial expansion of $\phi_{m,J}$ in terms of monomials $z_1^az_2^b$. It is then clear that $P_{\text{root}}\phi_{m,J}$ has a pattern of the form
\begin{equation}
  \dots 00100\dotsc 00100\dotsc,
\end{equation}
with two $1$'s separated by $m$ sites ($m$+1  sites) for $J$ odd ($J$ even), where the indices of the two occupied sites add up to $J$. Since $m\geq3$ for zero modes, root states of these 2-particle zero modes cannot have patterns of the form $\dotsc 0110\dotsc$ or $\dotsc 01010\dotsc$. That is, all such root states are annihilated by
$Z^1_j=c_j c_{j+1}$ and $Z^2_j=c_j c_{j+2}$, for any $j$.
Thus far, this holds for zero modes of the form $\phi_{m,J}$, $m\geq 3$. The theorem requires this to be true for {\em any} 2-particle zero mode, and thus for arbitrary linear combinations of $\phi_{m,J}$'s with $m\geq 3$.
One easily sees this to be the case: for different total angular momenta $J$, $J'$, a superposition of two 2-particle zero modes having these angular momenta results in the same superposition for the corresponding root states. When linearly combining $\phi_{m,J}$'s with different $m$ but identical $J$, the new root state is that corresponding to the $\phi_{m,J}$ with the largest $m$. Either way, the given $Z$-operators annihilate the root states of any such linear combination. By constructions, this is tantamount to the absence of sequences $11$ and $101$ in any root pattern, at first for 2-particle zero modes, but, by the theorem, also for $N$-particle zero modes. This, then, gives exactly the GPP for the $1/3$ Laughlin state, namely that at root level, there can be no more than one particle in any three adjacent sites. More precisely, this gives the GPP for {\em all} zero modes of the $V_1$-pseudo-potential. 

Note that this logic does not yet imply that for any pattern satisfying this rule, there is an associated zero mode. However, knowing that the Laughlin state is a zero mode and equation \eqref{eq:psiL} does have the pattern \eqref{eq:pattern} as its root pattern, and that no denser pattern can satisfy the GPP, this automatically proves the Laughlin state to be a densest possible zero mode of \eqref{eq:V1}.
We can also show that no other zero mode has the same angular momentum as the Laughlin state: given that the pattern 
\eqref{eq:pattern} is the unique pattern satisfying the GPP at this angular momentum, and given that any two zero modes would admit a non-trivial linear combination that is free of this pattern, we'd have a contradiction to the GPP. 
Similarly, knowing that a set of $N$-particle zero modes exists for every possible pattern consistent with the GPP, with some care, leads to the conclusion that this set is a complete set of zero modes\cite{Mazaheri2015,Chen2017, Bandyopadhyay2018,  Bandyopadhyay2019, 2022arXiv220409684T}. 

Though not explicitly using our theorem as stated here, a variant of this logic has first been presented for the Laughlin state in Ref. \cite{Mazaheri2015}. To be sure, for the Laughlin state and many other single component wave functions\cite{Stone1990, Read1996, Milovanovic1996,Schossler2021e}, there exist alternative ways to establish completeness of a class of zero mode wave functions, especially in the context of Jack polynomials\cite{Feigin2002, Bernevig2008, Bernevig2008a}. These alternatives are largely lost when multi-component wave functions are concerned, especially in the mixed Landau-level setting suitable for parton states.
For the latter, the methods following the philosophy laid out in this section have been successfully generalized,
\cite{Chen2017, Bandyopadhyay2018,  Bandyopadhyay2019, 2022arXiv220409684T}
 and the theorem presented here may serve to streamline this approach. One key additional feature of this approach when applied to multi-component states is that it will generally lead to so-called {\em entangled Pauli principles} (EPPs)\cite{Bandyopadhyay2018}: 
As the sum in equation \eqref{eq:Z} generally contains more than one term, equation \eqref{eq:ZZeroN} tends to enforce entanglement at root level.

\subsection{Jain-$2/5$ state}

A simple example of an EPP can be found in the (unprojected) Jain-2/5 composite fermion state. This state admits a parent Hamiltonian\cite{Jain1990,Rezayi1991} which acts within the lowest {\em two} Landau levels and enforces in any of its zero modes (of any particle number) the ``clustering condition'' that the wave function has at least a third-order zero when two particles approach the same point. In fact, when restricting to the lowest Landau level only, the zero mode condition of the $V_1$ Haldane pseudo-potential can also be stated in exactly this way, as is manifest in the 2-particle zero modes discussed above. It is, therefore, not surprising that this Hamiltonian can also be cast in the general form \eqref{eq:Hl} \cite{Chen2017}. In doing so, a non-orthogonal, non-normalized single particle basis $z^a \bar z^b$ is useful\cite{Chen2017, Bandyopadhyay2019}, where $a\geq 0$ and $b=0,1$ when restricting to the lowest two landau levels. Taking $j=a-b$ corresponds to physical angular momentum and leads to the scheme in figure  \ref{Fig1b} with $r=b$, while $j=a$ leads to the scheme in figure  \ref{Fig1a}, also with $r=b$. Note that each scheme comes with its own, slightly different definition of ``root state''. We will use the latter scheme (figure  \ref{Fig1a}) for definiteness.
As mentioned, the 2-particle zero modes discussed above remain zero modes for the two-Landau-level problem at hand. We also have the zero modes
\begin{equation} \label{eq:phimeven}
\tilde\phi_{m,J}=(z_1-z_2)^m\times
\begin{cases} 
  (z_1z_2)^{\frac{J-m+1}{2}}(\bar z_1-\bar z_2) & J\;\text{odd}   \\
  (z_1z_2)^{\frac{J-m}{2}}(z_1\bar z_1-z_2\bar z_2) & J\;\text{even}   
   \end{cases}\;,
\end{equation}
where $m$ is even, and $m\geq 2$, and 
\begin{equation} 
\tilde\phi_{m,J}=(z_1-z_2)^m\times
\begin{cases} 
  (z_1z_2)^{\frac{J-m}{2}}(\bar z_1z_1+z_2\bar z_2) & J\;\text{odd}   \\
  (z_1z_2)^{\frac{J-m+1}{2}}(\bar z_1+\bar z_2) & J\;\text{even}   
   \end{cases}\;,
\end{equation}
where $m$ is odd and $m\geq 3$. On top of that, there is a fourth category of zero modes, consisting of the original $\phi_{m,J}$ ($m\geq 3$) multiplied by $\bar z_1\bar z_2$.
In principle, as $r$ now runs over two values, two fermions could occupy states with the same $j$; at root level, however, all zero modes exhibit single occupancy in each orbital. Therefore, $Z^0_j=\tilde c_{j,0}\tilde c_{j,1}$ is now included in our set of $Z$-operators, where we introduced a tilde to refer to pseudo-fermions--fermion ladder operators that create/annihilate states in a basis that is not orthonormal. For this, we also define dual creation operators $\tilde c^\ast _{j,r}$ such that $\{\tilde c_{j,r},\tilde c^\ast_{j',r'}\}=\delta_{r,r'}\delta_{j,j'}$. We may work with pseudo-fermions just as we do with ordinary fermions, keeping in mind that $\tilde c^\ast _{j,r}$ and $\tilde c_{j,r}$ are not Hermitian conjugate, see, e.g., Ref. \cite{Bandyopadhyay2019}. 
Most importantly, we may re-express pseudo-fermions through ordinary fermions via a change of basis. Hence, our theorem admits $Z$-operators expressed in terms of pseudo-fermions without needing modification. Similarly, there is still no nearest-neighbor occupancy possible at root level, and we may introduce additional operators of the type $Z^1_j$ defined above, adorned with different combinations of $r$-indices, to formulate this constraint in two Landau levels. However, it is now possible, at root level, for particles to occupy next-nearest-neighboring sites. We see that this happens for $\tilde\phi_{2,J}$ with even $J$, whose root state is 
\begin{equation}\label{eq:singlet}
  \left(\tilde c^\ast_{J/2-1,0} \tilde c^\ast_{J/2+1,1} -     \tilde c^\ast_{J/2-1,1} \tilde c^\ast_{J/2+1,0}\right) \ket{0}\,.
\end{equation}
Conversely, this is the only 2-particle root state that can occur with two particles at distance 2. 
Therefore, when generalizing the operators $Z^2_j$ introduced above, we must form all possible linear combinations
$\sum_{r_1,r_2}\omega_{r_1,r_2} \tilde c_{J/2-1,r_1}\tilde c_{J/2+1,r_2}$
that annihilate equation \eqref{eq:singlet}. 
Note that if we associate a pseudo-spin $1/2$ degree of freedom with the $r$-label, then \eqref{eq:singlet} describes a singlet.
By the theorem, we are thus enforcing that for general $N$-particles zero modes, at root level, there are no double or nearest-neighbor occupancies, and next-nearest neighbors must form a pseudo-singlet. This, then, characterizes the EPP of the unprojected Jain-$2/5$ state and its parent Hamiltonian\cite{Chen2017}. It is easy to see that the densest root state consistent with these rules is that shown in figure  \ref{FigJain}. 
It is indeed a pattern that occupies $2/5$ of all $j$-sites, as we should expect of a fractional quantum Hall state at filling factor $2/5$. In particular, by ``monogamy of entanglement'', each particle in the root state can participate in only one singlet, and hence no denser root state is possible. Moreover, this is the root state of the Jain-$2/5$ wave function.
\begin{figure}
    \centering
    \includegraphics[scale=0.6]{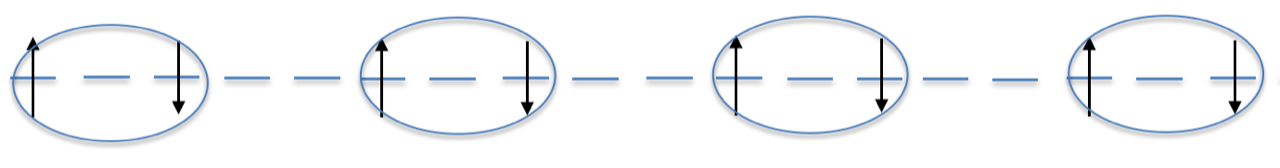}
    \caption{Root state, or DNA, of the Jain-2/5 wave function. Ovals represent singlets. This configuration manifestly occupies 2/5 of all available hall states, and satisfies the entangled Pauli principle that two particle entangled states must occur within a distance of 2. }
    \label{FigJain}
\end{figure}

We thus see how ``entangled DNA" can come about. By arguments similar to those given above for the Laughlin state, the Jain-$2/5$ state is seen to be the densest zero mode of its parent Hamiltonian.
In addition, it has been shown\cite{Chen2017, Chen2018b} that the counting of all possible root states at a given angular momentum agrees with the counting of zero modes, again utilizing similar arguments. In this way it can be rigorously established that this mode counting agrees with the counting of modes in the associated conformal edge theory.
Such reasoning is powerful in bridging the gap between microscopic Hamitlonians and effective field theory, and is not limited to Abelian fractional quantum Hall states such as Jain states, but rather has been carried out for a variety of increasingly complex non-Abelian ``parton'' states\cite{Bandyopadhyay2018, 2022arXiv220409684T}. Here, root-level entanglement becomes more non-local, even leading to the emergence of AKLT-type\cite{Affleck1987} matrix product entanglement and generalizations thereof. We refer the interested reader to the referenced literature.

\subsection{Tao-Thouless state}

We proceed by giving one final example of a single-component Hamiltonian that is perhaps less known.
We modify the pseudopotential $\hat V_1$ of  equation \eqref{eq:V12nd} by keeping only the terms with {\em odd} $J$. By abuse of terminology, we will denote the resulting potential as $\hat V_{\frac 12}$:
\begin{equation}
   \hat V_{\frac 12}  =\sum_{J>0, J\;\text{odd}} \,Q_J^\dagger Q_J\,,
\end{equation}
with $Q_J$ as defined in equation \eqref{eq:QJ}. We note that this operator, unlike $\hat V_1$, is no longer local in real space, but this is irrelevant for the purpose of this application.
We claim that the densest zero mode of $\hat V_{\frac 12}$ is the product state associated with the pattern $10101010\dotsc$, i.e., the 
``$\nu=1/2$ Tao-Thouless state''
\begin{equation}\label{eq:TT12}
  \ket{TT_{\frac 12}}= \ket{10101010\dotsc}
  \,.
\end{equation}
Indeed, it is clear that $\ket{TT_{\frac 12}}$ is a zero mode of $\hat V_{\frac 12}$; since the leading orbital has $j=0$, there is no pair in $\ket{TT_{\frac 12}}$ with odd $J$, so $Q_J \ket{TT_{\frac 12}}=0$ for all odd $J$.
On the other hand, showing that $\ket{TT_{\frac 12}}$ is the (unique) densest zero mode of  $\hat V_{\frac 12}$ requires more thought. The theorem, though, makes this an easy task. Notice that the operator $Q_J^\dagger Q_J$ is precisely the 2-particle projector onto $\psi_{1,J}$, so that the 2-particle zero modes of $\hat V_{\frac 12}$ are those which we had for $\hat V_1$ in addition to the $\psi_{1,J}$ with $J$ {\em even}, or equivalently, the $\phi_{m,J}$ with those same indices.
Therefore, $Z^2_j$ no longer annihilates the root states of all 2-particle zero modes: in particular, it does not annihilate $P_{\text{root}} \,\phi_{1,J}$ for $J$ even.
Thus, the associated constraint prohibiting the pattern $101$ at root level is moot. $Z^1_j$, however, still satisfies the assumptions of the theorem, and consequentially the pattern $11$ is still ruled out at root level. Clearly, the pattern underlying $\ket{TT_{\frac 12}}$ is the unique densest pattern avoiding the $11$ configuration (note that $\ket{TT_{\frac 12}}$ is its own root state). Therefore, repeating the argument made for the $\nu=1/3$-Laughlin state and the $V_1$ pseudo-potential, we are able to conclude that $\ket{TT_{\frac 12}}$ is the unique densest zero mode of
$\hat V_{\frac 12}$.

\subsection{Thin Cylinder}

Finally, one may observe that once a complete set of $Z$-operators has been identified, the Hamiltonian
\begin{equation}\label{eq:HZ}
    H_{\text{root}} = \sum_Z Z^\dagger Z
\end{equation}
stabilizes all root states consistent with the GPP or EPP as zero modes. As these root states are, in many cases, found to be in one-to-one correspondence with the zero modes of some original Hamiltonian \eqref{eq:Hl}, one may think of $H_{\text{root}}$ as an effective or ``thin cylinder'' version of the Hamiltonian, since a Hamiltonian stabilizing root states will naturally emerge in the thin cylinder/thin torus limit \cite{Rezayi1994, Seidel2005,  Bergholtz2005, Seidel2006, Bergholtz2006}. Given the formal resemblance between Eqs. \eqref{eq:Hl} and \eqref{eq:HZ}, one may ask why we find it beneficial to effectively replace, for certain purposes, $Q$-operators with $Z$-operators. The reason for this is simple:  While the former are usually long ranged in the occupation number bases, see equation \eqref{eq:QJ}, and have matrix product ground states of infinite bond dimension, the latter will be short ranged and have product states or matrix product states of finite bond dimension as ground states, i.e., the root states. It is for this reason that  $Z$-operators encoding a GPP/EPP are formidable tools to explore the structure of zero mode spaces.

We also note that while with each operator $Z$ we may associate a positive (semi-definite) Hermitian operator $Z^\dagger Z$, the converse is also true, in the following sense: Each positive (semi-definite) Hermitian operator can be 
decomposed as a sum of the form $\sum_Z Z^\dagger Z$, with $Z$-operators defined in terms of strings of annihilation operators of suitable length as in equation \eqref{eq:Z}. Therefore, one may easily obtain an equivalent formulation of our theorem, with $Z$-operators replaced by positive (semi-definite) Hermitian operators.

After this brief discussion of some of its utility, we are ready to prove our theorem:

\section{Proof\label{proof}}

\begin{proof}

Our proof strategy is technical and direct: standard (anti-)commutations relations combined with the fact that $Z$ annihilates all $l$-particle zero-modes yield the desired result. The trick is to recognize that any $N>l$ particle zero-mode can be connected to an $l$-particle zero-mode through annihilation operators. For example, let $\ket{\psi}$ be an $N$-particle zero mode of equation \eqref{eq:Hl}, $N>l$, and define
a {\em reduced} $l$-particle state 

\begin{equation} \ket{\psi^l} :=   c_{\beta_{N-l}} \dots c_{\beta_1}\ket{\psi} \,, \label{eq:psil} \end{equation} 
where we assume that indices $\beta_i = (j_i,r_i)$, $1\leq i\leq N$, exist such that
\begin{equation} \ket{S} = c_{\beta_1}^\dagger \dots c_{\beta_{N-l}}^\dagger c^\dagger_{ {\beta}_{N-l+1}} \dots c^\dagger_{{\beta}_N} \ket{0} \label{eq:S2}\end{equation}
is a root configuration of $\ket{\psi}$ with pattern $\{n_j\}$. Under the circumstances given above, we will refer to $\check\beta= (\beta_1,\dotsc,\beta_{N-l})$, as a {\em partial root configuration} of $\ket{\psi}$, and the underlying pattern $\{\check n_j\}$, i.e., the pattern of the state $c_{\beta_1}^\dagger \dots c_{\beta_{N-l}}^\dagger\ket{0}$, a {\em partial root pattern}. Similarly, we refer to $\beta^{{Z}}=(\beta_{N-l+1},\dots,\beta_N)$ as a {\em complementary root configuration} of $\check\beta$ in $\ket{\psi}$, and to the underlying pattern $\{n^Z_j\}$ as a {\em complementary root pattern}. While the reason for the superscript ``$Z$'' will soon be apparent, note that the complementary root configuration/root pattern is itself a partial root configuration/root pattern. By definition, $\check\beta$ is a partial root configuration if and only if a complementary root configuration $\beta^{Z}$ exists, such that $\beta=(\beta_1,\dotsc,\beta_N)$ is a root configuration. 
Also, note that $n_j=\check n_j + n^{Z}_j$, and observe that there is nothing in our definitions that prevents $\check n_j$ and  $n^{Z}_j$ from being simultaneously nonzero, so that neither of their non-zero parts should necessarily be thought of as a sub-sequence of $\{n_j\}$.

Since $\ket{S}$ appears in the expansion of $\ket{\psi}$, the complementary root configuration $\beta^{Z}$, 
\[  c_{\beta_{N-l}} \dots c_{\beta_1}\ket{S}= c^\dagger_{\beta_{N-l+1}} \dots c^\dagger_{\beta_N} \ket{0}, \]
 appears in the occupation-number expansion of
$\ket{\psi^l}$. The same is true of any other complementary root configuration for $\check\beta$.
Below, we will show that $\{n^{Z}_j\}$ is in fact a root pattern of $\ket{\psi_l}$, but for now we note that since in particular $\{n^{Z}_j\}$ is a pattern appearing in $\ket{\psi^l}$, $\ket{\psi^l}$ is non-vanishing.
 Moreover, as each $Q_J^{p,q}$ contains only annihilation operators, we have
\[ Q_J^{p,q}\ket{\psi^l} = \pm c_{\beta_{N-l}}\dots c_{\beta_1}Q_J^{p,q}\ket{\psi} = 0, \]
where the sign depends on $N-l$, $p$, and whether particles are bosons or fermions. It follows that $\ket{\psi^l}$ is an $l$-particle zero-mode.
If we want to emphasize the partial root configuration $\check \beta$ from which $\ket{\psi^l}$ was generated via equation \eqref{eq:psil}, we will write $\ket{\psi^l_{\check \beta}}$ instead.

By the assumptions of the theorem, $Z$ must annihilate the root state of $\ket{\psi^l}$.
We are particularly interested in the situation where the $\beta$-indices in the definition of $Z$ correspond to the pattern $\{n^{Z}_j \}$ in the preceding discussion. To make this precise, 
we change the indices in equation \eqref{eq:Z} via
\begin{align} \label{eq:Z2}
  Z &= \sum_{r_{N-l+1},\dots r_N} \omega_{r_{N-l+1},\dots,r_N}\, c_{\beta_{N-l+1}} \dots c_{\beta_{N}}\,.
\end{align}
Now we consider the sequence $\beta^{{Z}}=(\beta_{N-l+1},\dots,\beta_N)$ in the above expression, and again assume that the underlying pattern $\{n^{Z}_j \}$ (determined solely by the $j$-values) is a partial root pattern of $\ket{\psi}$ whose complement we denote by $\{\check n_j\}$.
Indeed, this assumption can be made without loss of generality.
Assuming the converse, i.e., that the $\{n^{Z}_j\}$ associated with $Z$ do not yield a partial root pattern of $\ket{\psi}$, is equivalent to the assumption that no root pattern $\{n_j\}$ of $\ket{\psi}$ satisfies $n_j\geq n^{Z} _j$ for all $j$. In that case, $Z$ already annihilates every root configuration of $\ket{\psi}$ trivially, due to insufficient particle number on (some of) the orbitals with the $j$-indices appearing in $Z$.

We now proceed by showing that, with 
$\ket{\psi^l}$ defined as above, $Z$
actually annihilates the entire state $\ket{\psi^l}$, and not just its root state. The key stepping stone is to show that, as anticipated,  $\{n^{Z}_j\}$ is a root pattern of $\ket{\psi^l}$.

\begin{lemma}\label{lem1}
Let $\ket{\psi}$ be an $N$-particle zero mode of $H$, equation \eqref{eq:Hl}. If $\{\check n_j\}$ is a partial root pattern of $\ket{\psi}$ and $\check\beta=(\beta_1,\dotsc,\beta_{n-l})$ a corresponding partial root configuration with a complementary root configuration $\beta^Z$, then $\beta^Z$ is a root configuration of $\ket{\psi^l_{\check\beta}}$.
In particular, the pattern $\{n^Z_j\}$
of $\beta^Z$ is a root pattern of $\ket{\psi^l_{\check\beta}}$.
\end{lemma}
\begin{proof}
We will prove this result by way of contradiction; suppose $\{n^Z_j\}$ is
not a root pattern of $\ket{\psi^l_{\check\beta}}$.
We have already established that it is a pattern of this state. For it not to be a root pattern, there must be another pattern $\{{n'}^Z_j\}$ of $\ket{\psi^l_{\check\beta}}$ that dominates  $\{{n}^Z_j\}$,  $\{{n'}^Z_j\}\geq \{{n}^Z_j\}$.
For  $\{{n'}^Z_j\}$ to be a pattern, there must be a state
\begin{equation}
    \ket{{\beta'}^Z} = c^\dagger_{\beta'_{N-l+1}}\dots c^\dagger_{\beta'_N} \ket{0}
\end{equation}
of pattern $\{{n'}^Z_j\}$ that appears in the expansion of $\ket{\psi^l_{\check\beta}}$ into configurations. This, in turn, is only possible if the state
\begin{equation}
  c^\dagger_{\beta_{1}}\dots c^\dagger_{\beta_{N-l}}   c^\dagger_{\beta'_{N-l+1}}\dots c^\dagger_{\beta'_N} \ket{0}
\end{equation}
appears in the full expansion into configurations of $\ket{\psi}$. This configuration clearly has the pattern $n'_j= \check n_j+ {n'}^Z_j$.
However, if $\{{n'}^Z_j\}\geq \{{n}^Z_j\}$, one immediately sees that
$\{{n'}_j\}\geq \{{n}_j\}$, where $n_j=\check n_j+n^Z_j$. 
For one, this follows easily from the fact that $\{n^Z_j\}$ is obtainable from  $\{n'^Z_j\}$ via inward squeezing processes, and the addition of $\{\check n_j\}$ to both does not change that. More generally, the addition of multiplicities $n_j=\check n_j+n^Z_j$
facilitates the ``union'' of the underlying partitions, $[j_i]= [\check j_i] \cup [j^Z_i]$, where the right hand side denotes the partition whose parts are those of $[\check j_i]$ and  $[j^Z_i]$ combined, after reordering. Similarly,  $[j'_i]= [\check j_i] \cup [j'^Z_i]$.  The statement then follows from the general fact that if $[\mu_i]$, $[\nu_i]$, $[\mu'_i]$, $[\nu'_i]$ are partitions with $[\mu'_i]\geq [\mu_i]$, $[\nu'_i]\geq [\nu_i]$, then also $[\mu'_i]\cup [\nu'_i]\geq [\mu_i]\cup [\nu_i]$ \cite{macdonald1998symmetric}. Note that $[\check j_i]\geq[\check j_i]$, as the dominance relation is reflexive.

Thus, since $\{{n'}_j\}\geq \{{n}_j\}$, $\{n_j\}$ is not a root pattern of $\ket{\psi}$, which contradicts the choices of $\{\check n_j\}$ and $\{n^Z_j\}$.
\end{proof}
The desired annihilation of $\ket{\psi^l}$ by $Z$ is then obtained as a corollary:
\begin{cor}\label{ZKillsL}
Let $\ket{\psi}$ be an $N$-particle zero mode of the Hamiltonian $H$, equation \eqref{eq:Hl}, and let $\{n^Z_j\}$ be the number of times $j$ appears among the $j_i$-indices in the operator $Z$ of Theorem \ref{rootthm}. Assume that $\{n^Z_j\}$ is a partial root pattern with complementary root pattern $\{\check n_j\}$. Let $\check \beta$ be a configuration with pattern $\{\check n_j\}$.
Then, under the assumptions of the Theorem,
\[Z \ket{\psi_{\check \beta}^l} = 0 \]
for the reduced state $\ket{\psi^l_{\check\beta}}$ of $\ket{\psi}$.
\end{cor}
\begin{proof}
There are two cases to consider: either $\check \beta =( \beta_1,\dotsc, \beta_{N-l})$ is a partial root configuration of $\ket{\psi}$ or that it is not. If it is, we write 
\begin{equation}
  \ket{\psi^l_{\check\beta}}  =\ket{\psi^l_{\check\beta}}_\text{root}+\ket{\text{rest}}\,,
\end{equation}
where the rest consists of that part of the expansion of $\ket{\psi^l_{\check\beta}}$ into occupation number eigenstates that does not consist of root configurations.
By the assumptions of the theorem, $Z$ annihilates $\ket{\psi^l_{\check\beta}}_\text{root}$.
But $Z$ also annihilates $\ket{\text{rest}}$, because, by Lemma \ref{lem1}, the pattern of indices in $Z$, $\{n_j^Z\}$, is a root pattern of $\ket{\psi^l_{\check\beta}}$, and $\ket{\text{rest}}$ contains no such pattern.
If $\check \beta$ is not a partial root configuration 
of $\ket{\psi}$, using equation \eqref{eq:Z2}, consider
\begin{equation}\label{eq:Zpsi}
   Z\ket{\psi_{\check \beta}^l} = \sum_{r_{N-l+1},\dots r_N} \omega_{r_{N-l+1},\dots,r_N}\, c_{\beta_{N-l+1}} \dots c_{\beta_{N}}
   c_{\beta_{N-l}} \dots c_{\beta_1} \ket{\psi}\,.
\end{equation}
The indices $\beta_1\dots\beta_N$ on the right hand side define a pattern $n_j = \check n_j + n^Z_j$. $\check n_j$ and $n^Z_j$ being complementary partial root patterns, $n_j$ is by definition a root pattern. 
Then the ket 
\[
   c_{\beta_1}^\dagger \dots c_{\beta_{N-l}}^\dagger c_{\beta_{N-l+1}}^\dagger c_{\beta_{N}}^\dagger\ket{0}
\]
cannot appear in the configuration-expansion of $\ket{\psi}$ for any choice of $r_{N-l+1}\dots r_N$. That is so since the configuration $\beta=(\beta_1,\dots,\beta_N)$, whose pattern is the root pattern $\{n_j\}$, would then be a root configuration. Therefore, $\check \beta$ would be a partial root configuration, contrary to assumption. Thus, 
$ Z\ket{\psi_{\check \beta}^l} =0$ in this case also.


\end{proof}

With this key fact established, we are nearly ready to essay the final calculation. To this end, we pose the following lemma. 
\begin{lemma}\label{NumberOpCommRel}
Let $Z$ be as in equation \eqref{eq:Z} 
and fix $j$, $k\geq0$ arbitrarily. 
Then
\[ \sum_{r_1,\dots,r_k} c^\dagger_{j,r_k}\dots c^\dagger_{j,r_1} Z c_{j,r_1}\dots c_{j,r_k} = (-1)^{lkf}Z \prod_{s = 0}^{k-1} (\hat n_j - ( n^Z_j +s)), \]
where $f$ is $1$ for Fermions and $0$ for Bosons,
and $n^Z_j$ is the number of occurrences of $j$ among the indices in $Z$, as before, and $\hat n_j$ is the number operator, equation \eqref{eq:number}.
\end{lemma}
 
\begin{proof}
Our proof will follow from induction applied to $k$. 
For simplicity, we write $m$ for $n^Z_j$, and employ the notation
\[ c_{j,r_1} \overset{i}{\overset{\uparrow}{\dots}} c_{j,r_n} = \prod_{\underset{1 \leq s \leq n}{s \neq i}} c_{j,r_s} \]  
to refer to a product with one element removed.
By linearity, it is sufficient to prove the theorem for operators of the form
\[
   \mathfrak{Z} = c_{j,r_1}\dotsc c_{j,r_{m}} c_{j_{m+1},r_{m+1}}\dotsc c_{j_l,r_l}\,,
\]
instead of $Z$, with $j\neq j_i$ for $i>m$.
Further we use the fact that 
\[ 
c^\dagger_x c_{y_1}\dots c_{y_n} = 
  (-1)^{nf}c_{y_1}\dots c_{y_n}c_x^\dagger - 
  \sum_{p=1}^n (-1)^{pf} \, \delta_{x,y_p} 
  c_{y_1}\overset{c_{y_p}}
  {\overset{\uparrow}{\dots}} c_{y_n} 
\] 
obtained from standard canonical commutation relations. Then,  
\begin{align*}
    \sum_r c^\dagger_{j,r} \mathfrak{Z} c_{j,r} &= \sum_{r}c^\dagger_{j,r}\of{ c_{j,r_1}\dots c_{j,r_m}\dots c_{j_l,r_l}} c_{j,r}\\
    &= \sum_{r}\of{(-1)^{mf}c_{j,r_1}\dots c_{j,r_m}c^\dagger_{j,r} - \sum_{p=1}^m (-1)^{pf}\delta_{r,r_p} c_{j,r_1}\overset{p}{\overset{\uparrow}{\dots}} c_{j,r_m}}c_{j_{m+1},r_{m+1}}\dots c_{j_l,r_l}c_{j,r} \\
    &= (-1)^{mf}c_{j,r_1}\dots c_{j,r_m} \sum_r c^\dagger_{j,r} c_{j_{m+1},r_{m+1}} \dots c_{j_l,r_l} c_{j,r} \\
    & \ \ - (-1)^{pf+(l-p)f}\sum_{p=1}^m c_{j,r_1} \dots c_{j,r_{p-1}}\of{\sum_r \delta_{r,r_p} c_{j,r}}c_{j,r_{p+1}}\dots c_{j_l,r_l} \\
    &= (-1)^{lf}c_{j,r_1},\dots c_{j_l,r_l} \of{\sum_r c^\dagger_{j,r} c_{j,r}} - (-1)^{lf}\sum_{p=1}^m c_{j,r_1}\dots c_{j_l,r_l} \\
    &= (-1)^{lf}\mathfrak{Z}\of{\hat n_j - m}.
\end{align*}
This proves the base case. For the induction step, we also note that
\[
    \hat{n}_j c_{j,r} = c_{j,r}\of{\hat{n}_j - 1}
\]
for both bosons and fermions. 
Now suppose the result holds for some $k$, and compute
\begin{align*}
    \sum_{r_1,\dots,r_{k+1}}& c^\dagger_{j,r_{k+1}}\dots c^\dagger_{j,r_1} \mathfrak{Z} c_{j,r_1}\dots c_{j,r_{k+1}}  \\
    & =\sum_{r_{k+1}} c_{j,r_{k+1}}  \of{ \sum_{r_1,\dots,r_k} c^\dagger_{j,r_k}\dots c^\dagger_{j,r_1} \,\mathfrak{Z}\, c_{j,r_1}\dots c_{j,r_k}} c_{j,r_{k+1}} \\
    &= \sum_{r_{k+1}} c_{j,r_{k+1}} ^\dagger \of{(-1)^{lkf}\, \mathfrak{Z} \prod_{s = 0}^{k-1} (\hat n_j - (m+s))} c_{j,r_{k+1}} \\
    &= (-1)^{lkf}\sum_{r_{k+1}} c_{j,r_{k+1}}^\dagger\,\mathfrak{Z} \,c_{j,r_{k+1}}  \of{  \prod_{s = 0}^{k-1} (\hat n_j - (m+s+1))}  \\
    &= (-1)^{l(k+1)f}\,\mathfrak{Z}\,(\hat n_j -m) \of{ \prod_{s = 0}^{k-1} (\hat n_j - (m+s+1))}\\
    &= (-1)^{l(k+1)f}\,\mathfrak{Z}\,\prod_{s=0}^k \of{ \hat n_j - (m+s)}, 
\end{align*}
as desired. This completes the induction. 
\end{proof}

\begin{cor}
\label{NumOpProd}
The formula 
\[ \prod_{s=0}^{k-1} \of{\hat n_j - s} = \sum_{r_1,\dots,r_k} c^\dagger_{j,r_1} \dots c^\dagger_{j,r_k}c_{j,r_k} \dots c_{j,r_1}. \]
holds for both bosons and fermions. 
\end{cor}
\begin{proof}
This is an application of lemma \ref{NumberOpCommRel} with $Z \equiv 1$. 
\end{proof}
Having gathered our pieces, all that remains is putting them together. For now, let $\ket{\psi}$ be a zero mode of the Hamiltonian \eqref{eq:Hl} with well-defined particle number $N$.
It suffices to show $Z\,P_{\{n_j\}}\ket{\psi} = 0$ for all root patterns
$\{n_j\}$ of $\ket{\psi}$.
Let $\{n_j^Z\}$ be the pattern associated with $Z$ as in the above. If $\{n_j^Z\}$ is not a partial root pattern with $n_j\geq n_j^Z$ for all $j$, $Z\,P_{\{n_j\}}\ket{\psi} = 0$ trivially, for the reasons mentioned before Lemma \ref{lem1}. Thus, let $\{n_j^Z\}$ be a partial root pattern with $n_j\geq n_j^Z$, such that 
$\check n_j = n_j-n_j^Z$ is a complementary root pattern. 
 If now $\check \beta=(\beta_1,\dots,\beta_{N-l})$ is any  configuration with pattern $\check n_j$, Corollary \ref{ZKillsL} tells us that 
\begin{align*}
    0 = Z \ket{\psi^l_{\check \beta}}
        = Z c_{\beta_{N-l}} \dots c_{\beta_1} \ket{\psi}.
\end{align*}
Next, consider the expansion of $\ket{\psi}$ into occupation-number eigenstates,
\[ \ket{\psi} = \sum_{\ket{S}} \kappa_{\ket{S}} \ket{S}\,. \]
Here, all $\ket{S}$ are $N$-particle configurations.
By considering $Z\ket{\psi^l_{\check \beta}}$
in the form \eqref{eq:Zpsi}, 
we see that all strings of annihilation operators in this expression have indices belonging to the pattern $\{n_j\}$. Therefore,
all $\ket{S}$ having a different pattern are trivially annihilated by the expression.
The expression thus remains valid when $\ket{\psi}$ is replaced with $P_{\{n_j\}}\ket{\psi}$:
\begin{align*}
    0 
        = Z c_{\beta_{N-l}} \dots c_{\beta_1} P_{\{n_j\}}\ket{\psi}.
\end{align*}
Acting on the above with $c_{\beta_1}^\dagger \dots c_{\beta_{N-l}}^\dagger$, and summing over the $r_i$'s hidden in the $\beta_i$'s, we have
\[  0 = \sum_{r_1,\dots,r_{N-l}} c_{\beta_1}^\dagger \dots c_{\beta_{N-l}}^\dagger Z c_{\beta_{N-l}} \dots c_{\beta_1} P_{\{n_j\}} \ket{\psi}.  \]
Applying Lemma \ref{NumberOpCommRel}, recalling that $\check n_j $ is the number
of times the index $j$ occurs in $\check \beta$, and similarly that
$n_j^Z $ is the number of occurrences of $j$ among the indies in $Z$, we obtain
\begin{align} 0 =  Z \,\prod_{\substack {j\\\check n_j>0}} (-1)^{l\check n_j f}   \of{\prod_{s=0}^{\check n_j-1} \of{\hat n_{j}-(n^Z_j+s)}}  P_{\{n_j\}}\ket{\psi} \label{eq:Zn}\,.
\end{align}
 We can then make the replacement
\begin{equation} P_{\{n_j\}}\ket{\psi} \rightarrow \prod_{\substack {j\\\check n_j>0}} \frac{\hat n_{j} \dots \of{ \hat n_{j}- (n_j^Z-1)}}{ n_{j} \dots \of{ n_{j}- (n_j^Z-1)}}
   P_{\{n_j\}} \ket{\psi}\,,\label{eq:insert} \end{equation}
   as the above operator product acts as the identity within the range of $P_{\{n_j\}}$. In particular, since
   $n_j\geq n^Z_j$, there are no singularities introduces in the denominator.
 Using this in equation \eqref{eq:Zn}, and dropping the    the non-zero c-number from the denominators in equation \eqref{eq:insert}, as well as the factors of $-1$ in equation \eqref{eq:Zn}, we obtain
\begin{align} 0 =  Z \,\prod_{\substack {j\\\check n_j>0}}  \of{\prod_{s=0}^{n_j-1} \of{\hat n_{j}-s}}  P_{\{n_j\}}\ket{\psi} \,,
\end{align} 
where the inner products now go up to $s=n_j-1$, rather than $s=\check n_j-1$. With the help of Corollary \ref{NumOpProd}, we may rearrange this to read
\[ 0 = Z \sum_{k_1,\dots,k_N} c^\dagger_{\beta_1}\dots c^\dagger_{\beta_{N}}c_{\beta_N} \dots c_{\beta_1} P_{\{n_j\}}\ket{\psi}. 
   \]
Finally, using equation \eqref{eq:Pn}, dropping further positive factors:
\[ 0 = Z P_{\{n_j\}} P_{\{n_j\}}\ket{\psi}
=Z  P_{\{n_j\}}\ket{\psi}
, \]  
and since $\{n_j\}$ was chosen arbitrarily, we have 
\begin{equation} Z P_{\text{root}} \ket{\psi} =Z\ket{\psi}_{\sf root}= 0, \label{eq:final}\end{equation}
as well. 
To complete the proof, we finally consider a zero mode $\ket{\psi}$ that is not necessarily a particle number eigenstate.
We emphasized that $(\ket{\psi_1}+\dotsc + \ket{\psi_m})_{\sf root}$ is not necessarily equal to $(\ket{\psi_1})_{\sf root}+\dotsc + (\ket{\psi_m})_{\sf root}$. However, if the components $\ket{\psi_1},\dotsc, \ket{\psi_m}$ have disjoint decompositions into configurations, i.e., are each composed of mutually disjoint subsets of configurations, then equality trivially holds.
This is in particular the case when $\ket{\psi}=\ket{\psi_1}+\dotsc+ \ket{\psi_m}$ is the decomposition of $\ket{\psi}$ into particle number eigenstates $\ket{\psi_i}$. 
equation \eqref{eq:final} then trivially generalizes to such $\ket{\psi}$.
This completes the proof of the root theorem. 
\end{proof}
\section{Conclusions\label{conclusion}}

In this paper, we have proved a general theorem about the root states of zero modes of large classes of positive-definite $k$-body Hamiltonians typical of fractional quantum Hall and related systems.
The theorem streamlines a methodology established in the literature to make rigorous statements about zero-mode spaces
of such Hamiltonians, leading to 
powerful constraints whenever
these spaces are non-trivial (which may or may not be the case for arbitrary particle number $N$), i.e., whenever the Hamiltonian is frustration free.
This includes a variety of situations where the Landau-level guiding-center degrees of freedom may coexist with others, particularly dynamical momenta that render first-quantized many-body wave functions non-holomorphic. The theorem represents a generalization of the methodology to $k$-body operators, where, at least in the non-holomorphic case, only $2$-body operators seem to have been considered previously. Moreover, while the very notion of root-states is naturally second-quantized, our theorem makes it easy to carry out the analysis in a mixed ``first-second'' quantized manner, while only dealing with states of finite particle number, as we have demonstrated in some examples. 
In contrast, the method used so far in the literature on parent Hamiltonians for non-holomorphic states embraces a slightly more cumbersome approach, where Pauli-like constraints on root states are derived with the particle number treated as an arbitrary unknown from the outset, and constraints on root states are derived using the fully second-quantized operators \eqref{eq:Q}.
Note that the examples we discussed indicate that the second quantized form \eqref{eq:Q} may not be needed {
\em explicitly} at all.
We expect that the theorem we proved in this work will significantly simplify the workflow when studying related problems, especially for $k>2$-body interactions and associated Pauli-like principles, and/or when root patterns occur that have some of the $n_j$ greater than two. We are thus hopeful that the present work will lead to further exciting developments in the study of frustration-free models for fractional quantum Hall-like systems, especially for states with internal and/or multiple Landau-level degrees of freedom, such as parton states and their generalizations.

\section*{Funding}This work has been
support by the National Science Foundation under Grant No. DMR-2029401.

\acknowledgments{
AS is indebted to
S. Bandyopadhyay, L. Chen,  B. Nachtergaele, Z. Nussinov,
G. Ortiz,  M. Tanhayi Ahari, S. Warzel, and A. Young
for insightful discussions.
AS is especially grateful to Technische Universit\"at M\"unchen and F. Pollmann in particular for their hospitality while some of this work has been carried out. }

\bibliography{main}

\end{document}